\documentclass[letterpaper, 10 pt, conference]{ieeeconf}
\usepackage{graphicx}      % include this line if your document contains figures
\usepackage{amssymb, amsfonts, amsmath, }
\usepackage{bbm}
\usepackage{epsfig, latexsym, amsfonts, amssymb, graphicx}
\usepackage{tikz} %add back

\newtheorem{theorem}{Theorem}
\newtheorem{lemma}[theorem]{Lemma}
\newtheorem{corollary}[theorem]{Corollary}

\newtheorem{definition}[theorem]{Definition}

\title{\LARGE \bf
Information Flow for Security in Control Systems
}
\author{Sean Weerakkody~~~~~Bruno Sinopoli~~~~~Soummya Kar~~~~~Anupam Datta}

\begin{document}
\maketitle
\begin{abstract}

This paper considers the development of information flow analyses to support resilient design and active detection of adversaries in cyber physical systems (CPS). The area of CPS security, though well studied, suffers from fragmentation. In this paper, we consider control systems as an abstraction of CPS. Here, we extend the notion of information flow analysis, a well established set of methods developed in software security, to obtain a unified framework that captures and extends system theoretic results in control system security. In particular, we propose the Kullback Liebler (KL) divergence as a causal measure of information flow, which quantifies the effect of adversarial inputs on sensor outputs. We show that the proposed measure characterizes the resilience of control systems to specific attack strategies by relating the KL divergence to optimal detection techniques. We then relate information flows to stealthy attack scenarios where an adversary can bypass detection. Finally, this article examines active detection mechanisms where a defender intelligently manipulates control inputs or the system itself in order to elicit information flows from an attacker's malicious behavior. In all previous cases, we demonstrate an ability to investigate and extend existing results by utilizing the proposed information flow analyses.
\end{abstract}
\section{Introduction}
The security of cyber physical systems (CPS), which integrate sensing, communication, and control in physical spaces, has become a significant challenge in society \cite{challengessecurity,cardenas2009rethinking}. Because CPS now pervade our critical infrastructures including transportation, manufacturing, health care, and energy, and are often implemented using off the shelf components, they offer both motivation and opportunity for potential attackers. There exist precedence for attacks on CPS including Stuxnet \cite{Langner2013} and the Maroochy Shire incident \cite{Slay2008} .

The ability to  detect and characterize such attacks is paramount to the well being of CPS. In particular,  to deliver appropriate countermeasures for attacks on physical systems, the operator must be able to passively detect attacks in a timely manner as they occur. Moreover, the defender must understand the set of stealthy attacks to motivate resilient design and active detection.  Here, passive detection refers to the defender's use of available information to ascertain if the system is operating normally or under attack. Passive detection techniques against attacks in CPS have been well studied. For instance, traditional methods of fault detection \cite{jones_report,failuredetection} have been considered. However, such schemes are usually designed to deal with benign failures. Consequently, recent work has aimed to consider the detection of stealthy adversaries who perform integrity attacks on sensor measurements and control inputs \cite{Mo_robust_detection}, \cite{PasqualettiAttack}, \cite{sundaram_wirelesscontrol}. 

Despite this previous work, the detection of arbitrary attacks on CPS by adversaries with diverse information and capabilities has not been well categorized. In this article, we propose using information flows as a means to quantify the detectability of generic adversarial attack models. Information flows analysis is an establised set of tools in software security \cite{DenningD77}, which attempt to determine if the processes of one agent alter the processes of another agent. We intend to use information flow to develop a unified treatment of security in CPS, specifically focusing on dynamical control aspects in this paper while leaving general cyber-physical treatments to future work. 

In this article, we propose the KL divergence as a quantitative measure for information flow to determine the extent to which an attacker's inputs in a control system affect the system outputs. To complement this measure, we introduce notions of conditional $\epsilon$-weak information flows and conditional $\epsilon$-strong information flows. Here, conditional weak information flows characterize stealthy attack strategies conditioned on the system model and the defender's control policy. Moreover, conditional strong information flows define active defense strategies which enable detection of adversarial behavior, conditioned on the attacker's policy. The resulting framework and analysis allows us to recover, in a unified manner, a collection of prior results in cyber-physical control systems, obtained using different techniques, in a number of papers. Moreover, in certain cases our framework allows us to present refinements on existing results to reveal additional insights. We summarize these instances below.

First, in section \ref{sec:passive},  we directly leverage the results of \cite{Bai2015} to demonstrate that the KL divergence allows us to characterize optimal passive detectability by specifically relating this measure to the optimal decay rate of the probability of false alarm. Moreover, we show through residue analysis and information theoretic bounds that the KL divergence can, in many instances, be efficiently evaluated. 

Next, in section \ref{sec:stealthy}, we consider the study of conditional weak information flows where we additionally assume a defender chooses an arbitrary control policy. We show that an adversary can develop attacks which generate $0$ information flow if and only if the system is left invertible. This allows us to recover results applied by \cite{PasqualettiAttack} to analyze undetectable attack scenarios. In addition, we show, under certain constraints on adversarial policy, that the information flow is a quadratic function of the bias injected on  measurement residues. This allow us to recover results in \cite{moscs10security} and \cite{mo2012integrity} on false data injections which used the residue bias as a constraint when studying impacts of stealthy adversaries. We are able to refine these results by presenting optimal detection guarantees for adversaries that satisfy these constraints.

Finally, information flow analysis allows us to consider results in active detection where the defender changes system parameters \cite{Tex2012}, \cite{miao_coding}, \cite{WeerakkodyCDC2015} or the control policy itself \cite{Mo2009R,Chabukswar2014,Weerakkody2014,Mo2014CSM,miao2013} to detect an attack. We consider the specific case of replay attacks. Here, we are able to recover results which show that certain systems and control policies are vulnerable to replay attacks \cite{Mo2009R}.  However, unlike \cite{Mo2009R} which uses specific continuity arguments, we use our framework to demonstrate that replay attacks generate a conditional weak information flow. We then recover results which state that introducing physical watermarking to the defender's policy \cite{Chabukswar2014} enables detection of replay adversaries. We do this by directly proving such a policy yields a conditional strong information flow for replay attacks. We are able to extend previous results \cite{Chabukswar2014} by using the calculated information flow to directly evaluate the detectability of a replay attack in a system with physical watermarking.

To close, we note that \cite{Bai2015} also leverages results relating the KL divergence to optimal passive detectability in order to define the notion of an $\epsilon$-stealthy attack. This is subsequently used to analyze maximum estimation degradation by a stealthy adversary in a scalar system. Our paper proposes using the KL divergence not only as a tool to analyze specific attacks, but as a unifying measure to characterize attacks and defenses in control system security. We also argue that our proposed framework is more general. Specifically, the notion of conditional information flow allows us to both characterize how an adversarial policy can be tuned to avoid detection by specific defenders and consider how the defender can adjust the system or his control policy to actively detect an attacker. We will revisit \cite{Bai2015} in a more technical context later.

The rest of the paper is summarized as follows. In section \ref{sysmodel}, we describe the system model. In section \ref{attackmodel}, we introduce a general model of an adversary in a CPS. Next, in section \ref{sec:infoflow} we define an information flow in a CPS through the KL divergence and relate it to existing notions in software security. After, in section \ref{sec:passive}, we motivation information flow as a computable measure of optimal passive detectability. In section \ref{sec:stealthy}, we discuss stealthy attack scenarios. Then, in section \ref{sec:active}, we consider information flow in the context of active detection. We conclude the paper in section \ref{sec:conclusion}.

\section{System Model} \label{sysmodel}
We consider a control system with discrete linear time invariant model given below.
\begin{equation}
x_{k+1} = Ax_k + Bu_k + w_k,~~y_k = Cx_k + v_k. \label{eq:dynamics}
\end{equation}

Here $x_k \in \mathbb{R}^n$ is the state, $u_k \in \mathbb{R}^p$ is the set of control inputs and $y_k \in \mathbb{R}^m$ is the set of sensor outputs. We let $x_0$ be the initial state. Furthermore, $w_k \sim \mathcal{N}(0,Q)$ and $v_k \sim \mathcal{N}(0,R)$ are independent and identically distributed (IID) process and IID measurement noise respectively. We consider a finite horizon up to time $T$. 

The previous linear model of a system is leveraged to derive the ensuing results related to control system security. However, we stress that the paradigm of information flows, to be introduced, can consider general nonlinear and time varying dynamical systems. 

We let $\mathcal{I}_k$ be the information available to the defender at time $k$ after making a measurement. From the defender's perspective, the initial state is unknown. However, the defender knows that $f(x_0|\mathcal{I}_{-1}) = \mathcal{N}(\hat{x}_{0|-1},P_{0|-1})$.  The defender at time $-1$ is aware of the system model $\mathcal{M} = \{ A, B, C, Q, R, \hat{x}_{0|-1},P_{0|-1} \}$. In total the defender's information at time $k$ is given by 
\begin{equation}
\mathcal{I}_k = \{ y_{0:k}, u_{0:k-1}, \mathcal{M} \}.
\end{equation}
$y_{0:k}$ refers to the finite sequence $\{y_0, \cdots, y_k \}$. Therefore, the defender is a central entity having cumulative knowledge of the dynamics of the system and the history of outputs and inputs. We now define an admissible defender control strategy as follows.

\begin{definition}
An admissible defender control strategy is a sequence of deterministic measureable functions $ \{\mathcal{U}_0, \mathcal{U}_{1}, \cdots, \mathcal{U}_{T-1}\}$ where $\mathcal{U}_k : \mathcal{I}_k \rightarrow \mathbb{R}^p$ for all $k \in \{0,1, \cdots, T-1 \}$ and $u_k = \mathcal{U}_k(\mathcal{I}_k)$.
\end{definition}

As a result, the defender computes a deterministic function of the current information to generate an input. 
%While we restrict $\mathcal{U}_k$ to be deterministic functions of the information available to the defender, they can include stochastic components, under the assumption that all stochastic components of the input are a priori known to the defender. 
Finally, we assume that the defender implements some passive bad data detector to determine whether the system is operating normally, denoted by a null hypothesis $\mathcal{H}_0$, or if there exist an abnormality (or possible attack), denoted by a state of $\mathcal{H}_1$. We define an admissible detector as follows.
\begin{definition}
An admissible defender detector strategy is a sequence of deterministic measureable functions $\{\Psi_0, \Psi_{1}, \cdots, \Psi_{T-1}\}$ where $\Psi_k : \mathcal{I}_k \rightarrow \{\mathcal{H}_0,\mathcal{H}_1\}$ for all $k \in \{0,1, \cdots, T \}$.
\end{definition}

Thus at each time $k$, the defender intelligently constructs a function $\Psi_k$ which maps the defender's available information to a decision about the state of the system, whether it is operating normally or has faulty and/or malicious behavior. 
%We assume the defenders input $u_k$ 
%
%
%We assume that the defender's input $u_k$ can be formulated as follows
%\begin{equation}
%u_k = u_k^{FB}(\mathcal{I}_k) + u_k^{FF}(\mathcal{I}_{-1})
%\end{equation}
%where $u_k^{FB}(\mathcal{I}_k)$ is a deterministic measureable function of $\mathcal{I}_k$ and refers to the feedback portion of the input and $u_k^{FF}(\mathcal{I}_{-1})$ is a deterministic measureable function of $\mathcal{I}_{-1}$ and refers to the feedforward portion of the input. Note that while we restrict $u_k^{FB}$ and $u_k^{FF}$ to be deterministic functions, they can include stochastic components, under the assumption that all stochastic components of the input are a priori known to the defender. This is in fact the case for stochastic inputs computed by a random number generator since the inputs for all time are determined by the seed of the generator.

\section{Attack Model} \label{attackmodel}
We now introduce an adversarial environment where an attacker, depending on his capabilities, as well as knowledge of the system can manipulate control inputs or sensor measurements to degrade control and estimation performance. Here, we formulate an adversary's effect on a system by including additive attacker inputs $u_k^a$ and $d_k^a$ as follows.
\begin{align}
x_{k+1} &= A x_k + B u_k + B^au_k^a + w_k, \label{eq:attackdyn} \\
y_k &= C x_k + D^a d_k^a +  v_k. \label{eq:outputdyn}
\end{align}

$B^a$ characterizes the adversarial inputs, which could be a subset of actuators the attacker usurps from the defender, or his own inputs. Without loss of generality, we assume $B^a$ is full column rank. We assume the adversary can modify $m^\prime$ sensors, $\mathcal{S} = \{ \gamma_1,\cdots,\gamma_{m^\prime} \} \subseteq \{1, \cdots, m \}$. Therefore, we define $D^a \in \mathbb{R}^{m \times m^\prime}$ entrywise as
$D_{u,v}^a  =  \mathbf{1}_{u = \gamma_j, v = j}.$

It is assumed that $u_k^a \in \mathbb{R}^{p^{\prime}}$ and $d_k^a \in \mathbb{R}^{m^{\prime}}$ are unknown to the defender. Thus a defender can only measure an adversary's effect on a system through sensor readings. 

We represent the adversary's knowledge of the system at time $k$ as $\mathcal{I}_k^a$. Here, we assume at a minimum that $\{u_{0:k-1}^a, d_{0:k}^a\} \subset \mathcal{I}_k^a$. Thus, the adversary is aware of his own history. Moreover, the adversary may have the ability to read a subset of control inputs $u_k$ or sensor outputs $y_k$ from the defender. For instance, if the attacker can modify channels, he may also be able to intercept signals sent along these channels, thereby utilizing a man in the middle attack. The portion of inputs and outputs the attacker and defender can read are public and are denoted $u_k^{pu}, y_k^{pu}$. Finally, the adversary may have some imperfect prior knowledge of the plant $\hat{\mathcal{M}}$, the controller $\hat{\mathcal{C}}$, and the detector $\hat{\mathcal{D}}$. The adversary's information is 
\begin{equation}
\mathcal{I}_k^a = \{u_{0:k-1}^a, d_{0:k}^a, u_{0:k-1}^{pu}, y_{0:k}^{pu}, \hat{\mathcal{M}},\hat{\mathcal{C}},\hat{\mathcal{D}} \}.
\end{equation}

An admissible attack strategy leverages the attacker's information $\mathcal{I}_k^a$ to generate attack inputs for the system.
\begin{definition}
An admissible attack strategy on the plant is a sequence of deterministic measureable functions $\{\mathcal{U}_0^a, \mathcal{D}_0^a,  \cdots, \mathcal{U}_{T-1}^a, \mathcal{D}_{T-1}^a, \mathcal{D}_{T}^a \}$ where $\mathcal{U}_k^a : \mathcal{I}_k^a \rightarrow \mathbb{R}^{p \prime}$ for all $k \in \{0,1, \cdots, T-1 \}$ and $u_k^a = \mathcal{U}_k^a(\mathcal{I}_k^a)$. Additionally, $\mathcal{D}_k^a : \mathcal{I}_{k-1}^a \times y_k^{pu} \rightarrow \mathbb{R}^{m \prime}$ for all $k \in \{0,1, \cdots, T \}$    and $d_k^a = \mathcal{D}_k^a(\mathcal{I}_{k-1}^a, y_k^{pu})$.
\end{definition}

We note that while current state of the art adversarial models for control systems consider attackers who do not change their attack strategy, our model considers an attacker with the freedom to leverage all his information to construct an attack input.

\section{Information Flows in Physical Systems} \label{sec:infoflow}
In software security, an information flow exists from a private input to a public output if including the private input changes the behavior of the public output. We wish to extend this notion for adversarial inputs and sensor outputs of control systems. In this section we propose a means to quantify information flow to characterize the detectability of adversarial strategies.

We quantify the information flow through the KL divergence between the distribution of the output under attack and the distribution of the output under normal operation \cite{Cover:2006:EIT:1146355}. For definiteness, we assume that all discrete time stochastic processes of interest considered hereafter induce (joint) distributions on the path space that are absolutely continuous with respect to Lebesgue measure. Thus, they possess densities in the usual sense. The KL divergence between a distribution with probability  density function $p(x)$ and a distribution with probability density function $q(x)$ over a sample space $X$ is given by
\begin{equation}
D_{KL}( p(x) || q(x) ) = \int_X   \log\left(\frac{p(x)}{q(x)} \right) p(x) d x.
\end{equation}
The above definition can be generalized to probability measures \cite{kullback1968information}. The KL divergence has the following properties \cite{Cover:2006:EIT:1146355}.
\begin{enumerate}
\item $D_{KL}( p(x) || q(x) ) \ge 0$ .
\item $D_{KL}( p(x) || q(x) ) = 0$ if and only if $p(x) = q(x)$ almost everywhere.
\item $D_{KL}( p(x) || q(x) ) \neq D_{KL}( q(x) || p(x) )$.
\end{enumerate}
We now use the KL divergence to define information flows in a physical system. To begin, denote the conditional distribution of the output based on apriori information as follows.
\begin{equation*}
\mathbb{D}_{y_{0:k}}^{\mathcal{M}, \mathcal{U}_{0:k-1}, \mathcal{U}_{0:k-1}^a, \mathcal{D}_{0:k}^a} = f(y_{0:k} | \mathcal{I}_{-1}, \mathcal{U}_{0:k-1}, \mathcal{U}_{0:k-1}^a, \mathcal{D}_{0:k}^a).
\end{equation*}
\begin{definition}
The information flow from the attacker's inputs $(\mathcal{U}_{0:T-1}^a, \mathcal{D}_{0:T}^a)$  to the defender's outputs $y_{0:T}$ is 
\begin{equation*}
IF_{T}=\frac{1}{T+1} D_{KL}(\mathbb{D}_{y_{0:T}}^{\mathcal{M}, \mathcal{U}_{0:T-1}, \mathcal{U}_{0:T-1}^a, \mathcal{D}_{0:T}^a} ||\mathbb{D}_{y_{0:T}}^{\mathcal{M}, \mathcal{U}_{0:T-1},0,0}).
\end{equation*}
\end{definition}

The proposed definition of information flows has many desirable properties, which make it compatible with existing measures of information flow in cyber security. First, the KL divergence allows us to recover the property of noninterference \cite{GoguenM82} in deterministic systems and probabilististic noninterference \cite{VolpanoS99} in stochastic systems.  There exists interference from a high level user to a low level user if changing high level inputs changes low level outputs. 

In our model, the low level inputs are the defender's actions, the high level inputs are the attacker's actions, and the low level outputs are the defender's outputs $y_{0:k}$.  In a deterministic system, if an adversary's actions change the output $y_{0:k}$, the KL divergence is infinite, reflecting the fact that there is interference. However, if the output $y_{0:k}$ is the same when the system is operating normally and under attack, indicating noninterference, the KL divergence is 0. There exists probabilistic interference from a high level user  to a low level user if changing high level inputs measurably alters the distribution of low level outputs. $IF_T = 0$ if and only if there exists probabilistic noninterference. 

Finally, when there exists probabilistic interference, we would like to have a means to measure information flow. In software security, this is done through research in quantitative information flow. A majority of previous work in software security \cite{Smith09} have proposed associative measures of information flow such as mutual information. Associative measures of information flow, which quantify correlation, attempt to evaluate how much information is leaked by an input to the output and thus provide utility in privacy applications. 

The KL divergence however is a causal measure which directly determines how varying an attacker's inputs changes the distribution of public outputs. The extent to which an attacker's input changes the system output will mark the defender's ability to distinguish outputs under attack from outputs under normal operation and thus detect the presence of an adversary. While recent work in software security has begun to investigate causal measures of information flow for violation detection, to our knowledge, the ensuing results will be the first work applied to physical systems.

To close the section we attempt to categorize adversarial policies which generate information flows bounded above by $\epsilon$ when the defender implements control policies in a set $\mathbf{U}$ or has a system with model in $\mathbf{M}$.
\begin{definition}
A permissible attack  $(\mathcal{U}_{0:T-1}^a, \mathcal{D}_{0:T}^a)$ generates a $(\mathbf{M},\mathbf{U})$ conditional $\epsilon$- weak information flow if for all $ \mathcal{U}_{0:T-1} \in \mathbf{U}$ and for all $ \mathcal{M} \in \mathbf{M}$, $IF_T \le \epsilon$. \label{modelrobust}
\end{definition}
Several special cases which satisfy this definition have arisen in the literature. For instance, a replay attack, generates an information flow bounded above by $\epsilon$ only for certain classes of models $\mathbf{M}$ and strategies $\mathbf{U}$. Another special case is below.
\begin{definition}
An adversary generates a $\mathcal{M}$ conditional $\epsilon$- weak information flow if for a specific model $\mathcal{M}$, $IF_T \le \epsilon$, regardless of the defender's policy $\mathcal{U}_{T-1}$. 
\end{definition}
This special case, where we remove any constraints on the defender's policy, is equivalent to $\epsilon$-stealthiness in \cite{Bai2015} and contains false data injections and zero dynamic attacks which we consider in section \ref{sec:stealthy}. We now consider defender policies and system design which elicit information flows. 
\begin{definition}
A change in the system $\mathcal{M}$ or a permissible control policy $\mathcal{U}_{0:T-1}$ generates a $\mathbf{U}^a$  conditional $\epsilon$- strong information flow if for $(\mathcal{U}_{0:T-1}^a, \mathcal{D}_{0:T}^a) \in \mathbf{U}^a$, $IF_T \ge \epsilon$.
\end{definition}
The preceding definition characterizes active detection where an adversary changes system parameters or his control policy to create an information flow. We will examine this topic further in section \ref{sec:active}.

\section{Passive Detection} \label{sec:passive}
In this section we motivate the KL divergence as a tool to quantify the passive detectability of an adversary and evaluate the special case of $\mathcal{M}$ conditional $\epsilon$- weak information flows. Specifically, we show that this measure is directly related to the optimal decay rate for the probability of false alarm.  We now have the following result from \cite{Bai2015}.

\begin{theorem} 
 Let $0 < \delta < 1$. Define $\alpha_k$ the probability of false alarm and $\beta_k$ the probability of detection as follows
\begin{align*}
\alpha_k & \triangleq \mbox{Pr}\left(\Psi_k(\mathcal{I}_k) = \mathcal{H}_0 | \mathcal{H}_0 \right),~ \beta_k \triangleq \mbox{Pr}\left(\Psi_k(\mathcal{I}_k) = \mathcal{H}_1 | \mathcal{H}_1 \right).
\end{align*}
Suppose  $\underset{k \rightarrow \infty}{\limsup}~ IF_k \ge \epsilon$. Then there exists a detector $\Psi_k$ such that \\  $\beta_k \ge 1-\delta,~\forall k,~~ \underset{k \rightarrow \infty}{\limsup} -\frac{1}{k+1} \log(\alpha_k) \ge \epsilon$. \\
Alternatively, suppose additionally that the sequences generated by $y_{0:k}$ operating normally and under attack are ergodic. Suppose $\underset{k \rightarrow \infty}{\lim}~ IF_k \le \epsilon.$
Then for all detectors $\Psi_k$
\begin{equation*}
\beta_k \ge 1-\delta,~\forall k \implies \underset{k \rightarrow \infty}{\limsup} -\frac{1}{k+1} \log(\alpha_k) \le \epsilon.
\end{equation*}
\label{stein}
\end{theorem}
Based on Theorem \ref{stein}, the information flow is essentially equivalent to the optimal decay rate in the probability of false alarm and an adversary who generates an $\mathcal{M}$ conditional $\epsilon$-weak information flow will have false alarm rate bounded above by $\epsilon$. As a result, information flow allows us to generically evaluate and compare the detectability of different attack policies. However unlike other potential measures such as $\beta_k$, the KL divergence can be efficiently characterized. 

We note that it may be difficult to compute the KL divergence of the outputs $y_{0:T-1}$ directly. For instance, if a control policy includes nonlinear feedback, the Gaussian property of the output is destroyed, which likely removes the ability to obtain closed form distributions of the output. We can instead consider the normalized residue $z_k$, obtained from a Kalman filter \cite{Kalman1960}.
\begin{equation}
\hat{x}_{k+1|k} = A\hat{x}_{k|k} + B u_k, ~\hat{x}_{k|k} = (I - K_kC)\hat{x}_{k|k-1} + K_ky_k, \label{eq:stateestimate}
\end{equation}
\begin{equation*}
P_{k+1|k} = AP_{k|k-1}A^T + Q - AK_kCP_{k|k-1}A^T,
\end{equation*}
\begin{equation*}
K_k = P_{k|k-1}C^T(CP_{k|k-1}C^T+R)^{-1},
\end{equation*}
\begin{equation}
z_k = (CP_{k|k-1}C^T+R)^{-\frac{1}{2}}(y_k - C\hat{x}_{k|k-1}).  \label{eq:residue}
\end{equation}
The Kalman filter computes optimal state estimates $\hat{x}_{k|k-1}$ and $\hat{x}_k$ of $x_k$. The normalized residue $z_k$ is a normalized measure of the difference between the defender's outputs and the expected outputs derived from the state estimate. We now have the following result \cite{mehra1971innovations}.
\begin{lemma}
The set of residues $f(z_{0:k}|\mathcal{I}_{-1}) = \mathcal{N}(0,I)$ when the system is operating normally. Given fixed strategy $\mathcal{U}_{0:k-1}$ and $\hat{x}_{0|-1}$, $z_{0:k}$ is an invertible function of $y_{0:k}$.  \label{lem:residue} 
\end{lemma}
Because the residues and outputs are related by an invertible mapping, we can show their KL divergences are equal \cite{kullback1968information}.
\begin{theorem} \label{residue_output}
The KL divergence between sensor outputs and between residues are equivalent.
\begin{align*}
& D_{KL}(\mathbb{D}_{y_{0:T}}^{\mathcal{M}, \mathcal{U}_{0:T-1}, \mathcal{U}_{0:T-1}^a, \mathcal{D}_{0:T}^a} ||\mathbb{D}_{y_{0:T}}^{\mathcal{M}, \mathcal{U}_{0:T-1},0,0}) \\ &= D_{KL}(\mathbb{D}_{z_{0:T}}^{\mathcal{M}, \mathcal{U}_{0:T-1}, \mathcal{U}_{0:T-1}^a, \mathcal{D}_{0:T}^a} ||\mathbb{D}_{z_{0:T}}^{\mathcal{M}, \mathcal{U}_{0:T-1},0,0})
\end{align*}
\label{residueKL}
\end{theorem}

Due to theorem \ref{residue_output}, we can analyze the residues operating normally and under attack instead of the system output when computing the information flow. Residues under normal operation have a known zero-mean Gaussian distribution. If the distribution of the residue under attack remains Gaussian, a closed form solution exists for the KL divergence. The KL divergence between two Gaussian distributions $\mathcal{N}_1 = \mathcal{N}_1(\mu_1, \Sigma_1)$ and $\mathcal{N}_0 = \mathcal{N}_0(\mu_0, \Sigma_0)$ with $\mu_1 \in \mathbb{R}^l$ is \cite{Cover:2006:EIT:1146355}
\begin{align}
 D_{KL}(\mathcal{N}_1 || \mathcal{N}_0) &= -\frac{l}{2} + \frac{1}{2} \mbox{tr}(\Sigma_0^{-1}\Sigma_1) + \frac{1}{2}\log\det \left( \Sigma_0 \Sigma_1^{-1} \right) \nonumber \\ &+ \frac{1}{2}  (\mu_1 - \mu_0)^T\Sigma_0^{-1}(\mu_1 - \mu_0). \label{KLnormal} 
\end{align}

If the attacker's policy is independent of the defender's outputs, it is known that the distribution of residues under attack remain Gaussian. In general however, it may still be difficult to compute the KL divergence of $z_{0:k}$ since it is a growing sequence. Fortunately, we can leverage the independence of the residues to obtain the following bound.

\begin{theorem}
The information flow generated by an adversary can be lower bounded by the sum of the residue-based KL divergences generated at each time step.
\begin{align*}
 IF_T \ge  \sum_{k=0}^T  \frac{D_{KL}(\mathbb{D}_{z_{k}}^{\mathcal{M}, \mathcal{U}_{0:k-1}, \mathcal{U}_{0:k-1}^a, \mathcal{D}_{0:k}^a} ||\mathbb{D}_{z_{k}}^{\mathcal{M}, \mathcal{U}_{0:k-1},0,0})}{T+1}.
 \end{align*}
 \label{sumofKLS}
\end{theorem}
\begin{proof}

By Theorem \ref{residue_output} and  Bayes rule we know
\begin{align*}
IF_T=  \sum_{k=0}^T  \frac{D_{KL}(\mathbb{D}_{z_{k}|z_{0:k-1}}^{\mathcal{M}, \mathcal{U}_{0:k-1}, \mathcal{U}_{0:k-1}^a, \mathcal{D}_{0:k}^a} ||\mathbb{D}_{z_{k}}^{\mathcal{M}, \mathcal{U}_{0:k-1},0,0})}{T+1}.
\end{align*}
Thus, we observe 
\begin{align*}
IF_T  - IF_T^{LB} =  \sum_{k=0}^T \frac{I_{z_k,z_{0:k-1}}^{\mathcal{M}, \mathcal{U}_{0:k-1}, \mathcal{U}_{0:k-1}^a, \mathcal{D}_{0:k}^a}}{k+1}.
\end{align*}
where $IF^{LB}_T$ is the obtained lower bound and $ I_{z_k,z_{0:k-1}}$ is the mutual information  \cite{Cover:2006:EIT:1146355} which is nonnegative.
\end{proof}

Instead of computing the KL divergence of vectors $z_{0:k} \in R^{m{k}}$, which in general requires us to store and compute the determinant of a matrix in $\mathbb{R}^{mk \times mk}$, we can instead obtain a recursive lower bound by computing the sum of $T$ divergences for vectors $z_k \in \mathbb{R}^m$. Moreover, note that the gap between the lower bound and $IF_T$ is the scaled sum of mutual informations between $z_k$ and $z_{0:k-1}$ so that if attack residues are independent, the gap is 0.

\section{Stealthy Adversarial Behavior} \label{sec:stealthy}
We next describe attacks which generate $\mathcal{M}$ conditional $\epsilon$-weak information flows, where regardless of the defender's policy the attacker remains stealthy. Understanding these scenarios motivate resilient design of $\mathcal{M}$ and also allow us to capture and extend research on left invertibility and false data injection attacks. The first scenario we consider is when $\epsilon = 0$ where there exists probabilistic noninterference. 

Let $y_{0:T}^a$ denote outputs realized from the distribution under attack $\mathbb{D}_{y_{0:T}}^{\mathcal{M}, \mathcal{U}_{0:T-1}, \mathcal{U}_{0:T-1}^a, \mathcal{D}_{0:T}^a}$ and $y_{0:T}$ denote outputs realized from the normal system $\mathbb{D}_{y_{0:T}}^{\mathcal{M}, \mathcal{U}_{0:T-1},0,0}$. If $\mathcal{U}_{0:T-1} = 0$, then, due to the linearity of our model $\mathcal{M}$,
\begin{align}
y_{0:T}^a &= y_{0:T} + \Delta y_{0:T} (d_{0:T}^a,u_{0:T-1}^a), \label{eq:dely} \\
\Delta x_{k+1} &= A \Delta x_k + B^a u_k^a, ~ \Delta x_0 = 0, \label{eq:delx} \\
 \Delta y_k &= C \Delta x_k+ D^a d_k^a. \label{eq:delyy}
\end{align}

We now obtain the following result.
\begin{theorem}
A nonzero attack strategy $(\mathcal{U}_{0:T-1}^a, \mathcal{D}_{0:T}^a)$ generates a $\mathcal{M}$ conditional 0-weak information flow if and only if $\Delta y_{0:T} (d_{0:T}^a,u_{0:T-1}^a) = 0$ with probability 1. \label{0attack}
\end{theorem}
\begin{proof}
Suppose $ \Delta y_{0:T} (d_{0:T}^a,u_{0:T-1}^a) = 0$ with probability $1 - \epsilon$ where $\epsilon > 0$. Then for $\mathcal{U}_{0:T-1} = 0$, we have with probability $1- \epsilon$, $y_{0:T}^a \neq y_{0:T}$. Thus, the KL divergence is greater than 0. Now instead suppose $ \Delta y_{0:T} (d_{0:T}^a,u_{0:T-1}^a) = 0$ with probability $1$. From \eqref{eq:attackdyn} and \eqref{eq:outputdyn}, we observe that \eqref{eq:dely} holds if $ \Delta y_{0:T} (d_{0:T}^a,u_{0:T-1}^a) = 0$. This is based on the fact that the defender's control strategy will not change if the output does not change. Thus, if $\Delta y_{0:T} (d_{0:T}^a,u_{0:T-1}^a) = 0$ with probability $1$, then $y_{0:T}^a = y_{0:T}$ with probability 1. Therefore, the KL divergence and information flow is 0.
\end{proof}
We have shown that there exists a $0$-information flow attack if and only if there exists nontrivial $(\mathcal{U}_{0:T-1}^a, \mathcal{D}_{0:T}^a)$  which satisfy \eqref{eq:delx}, \eqref{eq:delyy} for $0 \le k \le T$. For long enough time horizon this is in fact equivalent to left invertibility.
\begin{theorem}
Let $\hat{B}^a = \begin{bmatrix} B^a & 0_{n \times m^\prime} \end{bmatrix},~\hat{D}^a = \begin{bmatrix} 0_{m \times p^\prime} & D^a \end{bmatrix}$. Suppose $T \ge n - p^{\prime} + 1$. A nonzero adversarial policy $(\mathcal{U}_{0:T-1}^a, \mathcal{D}_{0:T}^a)$ can generate a $\mathcal{M}$ conditional $0$-weak information flow if and only if $(A,\hat{B}^a,C,\hat{D}^a)$ is not left invertible.
\end{theorem}
\begin{proof}
The result follows directly from Theorem \ref{0attack} and  Corollary 1 of \cite{Willsky1974}.
\end{proof}
Left invertibility in control systems has been well studied in previous work in CPS security as a subset of zero dynamic attacks \cite{PasqualettiAttack}. Our general framework of information flows is able to recover this property and consequently, we can directly apply previous results related to left invertibility in our study of $0$-weak information flows. For instance, we can consider conditions on $\mathcal{M}$ which allow for the existence of 0 information flow attacks to motivate resilient design of the system $(A,B,C)$ and channel security $(B^a,D^a)$.
\begin{theorem}
\label{matrixpencil}
\cite{PasqualettiAttack}  Let $T \ge n - p^{\prime} + 1$. An attack policy can create a $\mathcal{M}$ conditional $0$-weak information flow if and only if $\mbox{rank} \left( \bar{P}(\mathcal{M}) \right) < n + p^\prime + m^\prime,  ~~~ \forall~\lambda \in \mathbb{C}$
\begin{align}
 \mbox{ where } \bar{P}(\mathcal{M}) =  \begin{bmatrix} \lambda I - A &  \hat{B}^a \\  C &  \hat{D}^a \end{bmatrix}.  \nonumber
\end{align}
\end{theorem}

We now wish to consider the case of $\mathcal{M}$ conditional $\epsilon$-weak information flows for $\epsilon > 0$. However, we assume that the adversary injects additive inputs which are independent of the defender's system outputs. Thus, we assume
\begin{align}
u_k^a &= \mathcal{U}_k^a(u_{0:k-1}^a,d_{0:k}^a,\hat{\mathcal{M}},\hat{\mathcal{C}},\hat{\mathcal{D}}), \nonumber \\
d_k^a &= \mathcal{D}_k^a(u_{0:k-1}^a,d_{0:k-1}^a,\hat{\mathcal{M}},\hat{\mathcal{C}},\hat{\mathcal{D}}). \label{eq:fdipolicy}
\end{align}
Such attacks are known as false data injection attacks. We now have the following result.
\begin{theorem}
\label{DELZ}
Consider an admissible adversarial policy which satisfies \eqref{eq:fdipolicy}. Then,
\begin{equation}
IF_T =  \frac{1}{2(T+1)} \Delta z_{0:T}^T \Delta z_{0:T},
\end{equation}
where $\Delta z_k$ satisfies $\Delta e_{0|-1} = 0$ and
\begin{align}
&\Delta e_{k+1|k} = (A-AK_kC) \Delta e_{k|k-1} + B^au_k^a - AK_kD^ad_k^a, \nonumber\\ 
&\Delta z_k = (CP_{k|k-1}C^T+R)^{-\frac{1}{2}} \left(C \Delta e_{k|k-1} + D^ad_k^a \right). \label{eq:delta_est}
\end{align}
\end{theorem}
\begin{proof}
See Appendix \ref{app:fd}.
\end{proof}
Thus, the information flow is proportional to the norm of $\Delta z_k$ squared where $\Delta z_k$ represents the bias the adversary injects on the normalized residue. The norm of the residue bias has been previously used as a measure of the stealthiness in false data injection attacks. For instance, \cite{moscs10security} and \cite{mo2012integrity}, in their investigation of false data injection attacks, restrict
\begin{equation}
\|\Delta z_k\|^2 \le B ~~\forall k.
\end{equation}
with the motivation that the increase in $\beta_k$ will be bounded by some $B^{\prime}$ in this scenario. For $B \le 2\epsilon$, such an attacker generates a $\mathcal{M}$ conditional $\epsilon$-weak information flow. Consequently we have the following result.
\begin{theorem}
Suppose a false data injection attack satisfies $\|\Delta z_k\|^2 \le 2\epsilon~ ~\forall k$. Then, for $\delta > 0$ there exists a detector such that $\beta_k \ge 1-\delta$ and $\underset{k \rightarrow \infty}{\limsup} -\frac{\log(\alpha_k)}{k+1} = \epsilon$.
\end{theorem}
Again, the results obtained in \cite{moscs10security}, evaluating models $\mathcal{M}$ and attacks $\mathcal{D}_{0:k}^a$ which stealthily destabilize a system, and \cite{mo2012integrity}, estimating the bias an adversary can stealthily inject on the system state in $\mathcal{M}$,  can all be reframed as  attacks which generate $\mathcal{M}$ conditional $\epsilon$-weak information flow. This refinement of existing results allows us to now quantify detectability in addition to system impact.

\section{Active Detection of Adversarial Behavior} \label{sec:active}
In this section, we will revisit and extend results related to the active detection of replay attacks using the proposed measure of information flow. Recall that in active detection, the defender changes the system or his policy to elicit an information flow. Specifically, we will use information flows to determine when replay attacks are stealthy. We will then extend previous work by using information flows to characterize optimal detection with watermarking. 

In a replay attack, the adversary observes a sequence of measurements from $y_{-N}$ to $y_{-N+T-1}$. Then, without loss of generality, at time $0$, the attacker replays these measurements. Here, we will assume $-N$ is large so that the adversary has an adequate buffer and that the replayed outputs are independent of the current outputs. Moreover we assume the system at time $-N$ is in steady state. We first argue that a replay attack generates a $(\mathbf{M},\mathbf{U})$ conditional $\epsilon$-weak information flow for a large class of systems $\mathbf{M}$ and common control policies $\mathbf{U}$. For instance, consider a defender that uses state feedback with gain $L$ so $\mathcal{U}_k(\mathcal{I}_k) = L\hat{x}_{k|k}$. 

%Let $\mathcal{A} = (A+BL)(I-KC)$. In this case, the closed loop system under normal operation is given by
%\begin{align}
%\begin{bmatrix} x_{k+1} \\ \hat{x}_{k+1|k} \end{bmatrix} &= \begin{bmatrix} A + BLKC & BL(I-KC) \\ (A+BL)KC & \mathcal{A} \end{bmatrix}\begin{bmatrix} x_{k} \\ \hat{x}_{k|k-1} \end{bmatrix}  \\ &+ \begin{bmatrix} I & BLK \\ (A+BL)K & 0 \end{bmatrix}\begin{bmatrix} w_k \\ v_k \end{bmatrix}, \nonumber \\
%y_k &= \begin{bmatrix} C & 0 \end{bmatrix} \begin{bmatrix} x_{k} \\ \hat{x}_{k|k-1} \end{bmatrix} + v_k. \label{eq:closedreplay}
%\end{align}
Let $\mathcal{A} = (A+BL)(I-KC)$  and $\mathcal{P} = CPC^T+R$. It has been shown that \cite{Chabukswar2014}
\begin{equation}
z_k = z_{k-N} - \mathcal{P}^{-\frac{1}{2}}C\mathcal{A}^k(\hat{x}_{0|-1} - \hat{x}_{-N|-N-1}). \label{eq:residuereplay}
\end{equation}
If $\mathcal{M}$ and $\mathcal{U}_{0:k-1}$ generate stable $\mathcal{A}$ the second term converges to 0. Therefore, we have the following result regarding the information flow with proof in appendix \ref{app:noflow}. 
\begin{theorem} 
 \label{infoflow_nowatermark}
Suppose that our control system \eqref{eq:dynamics} with state feedback control is under replay attack, where $\rho(\mathcal{A}) < 1$. Then, $\underset{T \rightarrow \infty}{\lim} IF_T = 0$.
\end{theorem}

If $\mathcal{A}(\mathcal{M},\mathcal{U}_{0:k-1})$ is stable, the adversary's actions are asymptotically undetectable since the information flow is 0. This result was previously obtained in \cite{Mo2009R} by instead showing that continuous functions of the defender's information are indistinguishable under normal and replay scenarios. Information flows allow us to recover this result via a general CPS security framework.

In this example, the defender's control strategy $\mathcal{U}_{0:T-1}$ of state feedback, leaves the system vulnerable to a replay attack. The defender ideally should be able to perform active detection and determine a control strategy which simultaneously addresses system objectives while creating an information flow from a replay adversary. 

Watermarking techniques allow the defender to increase the information flow from the attacker input to defender output and as a result create an $\mathbf{U}^a$ conditional $\epsilon$-strong information flow, where $\mathbf{U}^a$ contains the replay attack policy. In watermarking, noisy control inputs are used with $u_k = \mathcal{U}_k(\mathcal{I}_k) = L\hat{x}_{k|k} + \Delta u_k$ where $\Delta u_k \sim \mathcal{N}(0,\mathcal{Q})$. Note that while the watermark is random, it can be predetermined offline so that $\mathcal{U}_k(\mathcal{I}_k)$ remains a deterministic function. We now show watermarking creates a strong information flow.

\begin{theorem}
Suppose the system \eqref{eq:dynamics} with state feedback control and watermarking is under replay attack, where $\rho(\mathcal{A}) < 1$. Then, almost surely $\underset{T \rightarrow \infty}{\lim} IF_T \ge \epsilon$, where
\begin{equation*}
\epsilon = \dfrac{ \mbox{tr} \left(\mathcal{P}^{-1}C\Sigma C^T\right)}{2}, ~~~\Sigma  = \mathcal{A}\Sigma\mathcal{A}^T + B\mathcal{Q}B^T.
\end{equation*}
\label{infoflow_watermark}
\end{theorem}
\begin{proof}
See Appendix \ref{app:flow}.
\end{proof}

From the theorem above, the defender can make the information flow from an adversarial input arbitrarily large by increasing  $\mbox{tr} \left(\mathcal{P}^{-1}C\Sigma C^T\right)$ which is a linear function of the watermark covariance $\mathcal{Q}$. In fact, previous work on watermarking \cite{Chabukswar2014} does aim to design watermarks by maximizing $\mbox{tr} \left(\mathcal{P}^{-1}C\Sigma C^T\right)$ subject to constraints on control performance in the system. Thus, our results motivate the choice of this objective function. The use of information flows also allow us to extend previous results to analyze optimal detection of replay attacks under watermarking scenarios.
\begin{corollary}
Assume system \eqref{eq:dynamics} with state feedback control and watermarking is under replay attack, where $\rho(\mathcal{A}) < 1$. Then for $\delta > 0$ there exists a detector such that $\beta_k \ge 1 - \delta,~\forall~k$ and
\begin{equation}
\underset{k \rightarrow \infty}{\limsup} -\frac{1}{k} \log(\alpha_k) \ge \dfrac{ \mbox{tr} \left(\mathcal{P}^{-1}C\Sigma C^T\right)}{2}.
\end{equation}
\end{corollary}
\begin{proof}
The result follows from Theorems \ref{infoflow_watermark} and \ref{stein}.
\end{proof}

We simulate a vehicle moving along a single axis \cite{moscs10security} under replay attack. Here, we assume that the defender obtains the gain $L$ using a linear quadratic Gaussian (LQG) controller which attempts minimize a cost $J$ given by
\begin{displaymath}
J = \lim_{T \rightarrow \infty} \frac{1}{T+1} \mathbb{E} \left[ \sum_{k=0}^T x_k^Tx_k + u_k^Tu_k \right].
\end{displaymath}
The LQG cost increases linearly with $\mathcal{Q}$. We select the covariance $\mathcal{Q}$ of the watermark so that $\Delta J$, the increased cost due to watermarking, is $40 \%$ of the optimal $J$. Here, we simulate the system $1000$ times over a horizon of 200 steps. We plot the average information flow in Fig \ref{infoflow_car}, both with watermarking and without watermarking. As expected from Theorem \ref{infoflow_nowatermark}, in the absence of watermarking, the information flow generated by a replay attack converges to 0. If physical watermarking is implemented, the information flow generated by an adversary has a lower bound $\epsilon$ which grows linearly with $\mathcal{Q}$. We implement a Neyman Pearson detector \cite{Cover:2006:EIT:1146355} and plot the average probability of false alarm and detection as a function of $k$ in Fig \ref{detection_car}. 

\begin{figure}
\begin{center}
\includegraphics[scale=0.4]{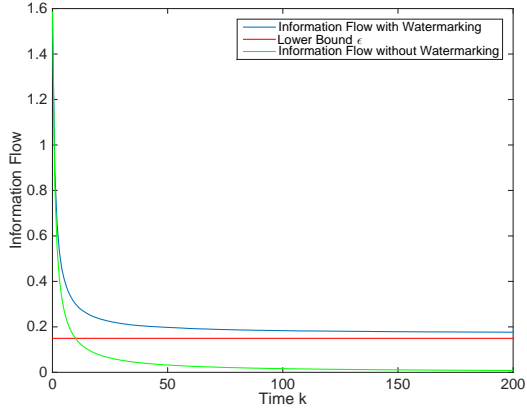}
\end{center}
\caption{Information Flow generated by a replay attack. The information flow as a function of $k$ in the presence of watermarking is included along with its lower bound $\epsilon$, and the information flow generated when physical watermarking is not present}
\label{infoflow_car}
\end{figure}

\begin{figure}
\begin{center}
\includegraphics[scale=0.4]{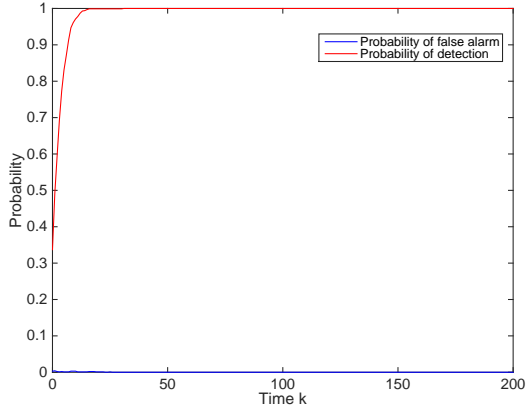}
\end{center}
\caption{Probability of detection and probability of false alarm vs time for a Neyman Pearson Detector}
\label{detection_car}
\end{figure}

\section{Conclusion} \label{sec:conclusion}
In this article, we introduced a physical measure of information flow to characterize detection in CPS and provide a unified approach to dealing with security in both the cyber and physical domains. We proposed the KL divergence as a measure of information flow. We motivate its use through results in optimal passive detection and computational ease of evaluation. We examined attacks which are stealthy for fixed models, and all input strategies, recovering results related to left invertibility and false data injection attacks. Finally, we investigated replay attacks and used information flows to quantify optimal detection performance with physical watermarking. We close by noting that information flow tools are amenable to true CPS analysis. In particular, we can consider a richer set of problems emcompassing both cyber and physical domains by leveraging the proposed results in physical security and existing parallels in cyber and software security. Approaching these general problems will mark the next stage of obtaining a unified paradigm for addressing CPS security.

\appendices

\section{Proof of Theorem \ref{DELZ}} \label{app:fd}
\begin{proof}
Let $e_{k|k-1} = x_k - \hat{x}_{k|k-1}$. From \eqref{eq:attackdyn},\eqref{eq:outputdyn}, and \eqref{eq:stateestimate} we obtain 
\begin{align*}
e_{k+1|k} =  (A-AK_kC) e_{k|k-1} + B^au_k^a &+ w_k  - AK_kv_k \\   &- AK_kD^ad_k^a, 
\end{align*}
\begin{equation*}
z_k = (CP_{k|k-1}C^T+R)^{-\frac{1}{2}} \left(C e_{k|k-1} + v_k +  D^ad_k^a \right).
\end{equation*}
Let $z_k^s$ be the residue under normal operation, where $\mathcal{U}_{0:k-1}^a = 0$ and $\mathcal{D}_{0:k}^a = 0$. Then,
\begin{equation*}
e_{k+1|k}^s =  (A-AK_kC) e_{k|k-1}^s  + w_k  - AK_kv_k,
\end{equation*}
\begin{equation*}
z_k^s = (CP_{k|k-1}C^T+R)^{-\frac{1}{2}} \left(C e_{k|k-1}^s + v_k  \right).
\end{equation*}
It can be seen from the linearity of the system that 
\begin{equation*}
z_k = z_k^s + \Delta z_k,
\end{equation*}
and that \eqref{eq:delta_est} holds. Moreover, from an inductive argument, we see that $\Delta z_k$ is a deterministic variable since  $\mathcal{U}_{0:k-1}^a$ and $\mathcal{D}_{0:k}^a$ are known functions of deterministic variables $\hat{\mathcal{M}}, \hat{\mathcal{C}}, \hat{\mathcal{D}}$. As a result, $\mathbb{D}_{z_{0:k}}^{\mathcal{M}, \mathcal{U}_{0:k-1}, \mathcal{U}_{0:k-1}^a, \mathcal{D}_{0:k}^a} = \mathcal{N}(\Delta z_{0:k}, I)$. Finally, from \eqref{KLnormal} and Theorem \ref{DELZ}, we have
\begin{equation*}
D_{KL}\left(\mathcal{N}(\mu_1,\Sigma_1),\mathcal{N}(\mu_2,\Sigma_1)\right) = \frac{1}{2} \| \Sigma_1^{-\frac{1}{2}} (\mu_1-\mu_2) \|^2.
\end{equation*}
The result immediately follows.
\end{proof}
\section{Proof of Theorem \ref{infoflow_nowatermark}} \label{app:noflow}
\begin{proof}
We observe from \eqref{eq:residuereplay} that
\begin{equation}
z_{0:k} \sim \mathcal{N}(\mu_r,\Sigma_r),
\end{equation}
\begin{equation}
\mu_r(jm:jm+m-1) = \mathbb{E}[z_j] = -\mathcal{P}^{-\frac{1}{2}}C\mathcal{A}^k\hat{x}_{0|-1},
\end{equation}
\begin{align}
\Sigma_r &(jm:jm+m-1,lm:lm+m-1) = \mbox{Cov}(z_j,z_l^T),  \nonumber \\ &= \mathcal{P}^{-\frac{1}{2}}C \mathcal{A}^j\mathcal{W}(\mathcal{A}^{l})^TC^T\mathcal{P}^{-\frac{1}{2}} + \delta(l-m)I,
\end{align}
where $\mathcal{W}$ is the steady state covariance of $\hat{x}_{k|k-1}$ and $\delta$ refers to the discrete delta dirac function. From \eqref{KLnormal}, Theorem \ref{residueKL}, and Sylvester's determinant theorem we have
\begin{equation*}
 D_{KL}(\mathbb{D}_{y_{0:k}}^{\mathcal{M}, \mathcal{U}_{0:k-1}, \mathcal{U}_{0:k-1}^a, \mathcal{D}_{0:k}^a} ||\mathbb{D}_{y_{0:k}}^{\mathcal{M}, \mathcal{U}_{0:k-1},0,0}) = \frac{c_1 + c_2 + c_3}{2}
\end{equation*}
where
\begin{align*}
c_1 &= \mbox{tr} \left( \sum_{j = 0}^k \mathcal{P}^{-\frac{1}{2}}C \mathcal{A}^j\mathcal{W}(\mathcal{A}^{j})^TC^T\mathcal{P}^{-\frac{1}{2}} \right), \\
c_2  &= \sum_{j = 0}^k \hat{x}_{0|-1}^T (\mathcal{A}^j)^TC^T\mathcal{P}^{-1}C\mathcal{A}^j\hat{x}_{0|-1},\\
c_3 &= -\log \det \left( I + \sum_{j=0}^k \mathcal{W}^{\frac{1}{2}} (\mathcal{A}^j)^TC^T\mathcal{P}^{-1}C\mathcal{A}^j \mathcal{W}^{\frac{1}{2}} \right).
\end{align*}
Let $X_1$ and $X_2$ be given by
\begin{equation*}
X_1 = \sum_{j = 0}^\infty  \mathcal{A}^j\mathcal{W}(\mathcal{A}^{j})^T = \mathcal{A} X_1 \mathcal{A}^T + W,
\end{equation*}
\begin{equation*}
X_2 = \sum_{j = 0}^\infty  (\mathcal{A}^j)^T C^T \mathcal{P}^{-1} C \mathcal{A}^{j} = \mathcal{A}^T X_2 \mathcal{A} +  C^T \mathcal{P}^{-1} C.
\end{equation*}
From Lyapunov's equation and since $\mathcal{A}$ is stable, the matrices $X_1$ and $X_2$ exist and are bounded. Since $c_1$, $c_2$, and $|c_3|$ are monotonic in $k$, we have for all $k$
\begin{equation*}
c_1 \le \mbox{tr} \left (\mathcal{P}^{-\frac{1}{2}} C X_1 C^T  \mathcal{P}^{-\frac{1}{2}}\right),~ c_2 \le  \hat{x}_{0|-1}^T X_2 \hat{x}_{0|-1}^T,
\end{equation*}
\begin{equation*}
|c_3| \le  \log \det \left(I + \mathcal{W}^{\frac{1}{2}} X_2 \mathcal{W}^{\frac{1}{2}} \right).
\end{equation*}
Consequently, for all $k$ there exists $M^*$ satisfying
\begin{equation*}
 D_{KL}(\mathbb{D}_{y_{0:k}}^{\mathcal{M}, \mathcal{U}_{0:k-1}, \mathcal{U}_{0:k-1}^a, \mathcal{D}_{0:k}^a} ||\mathbb{D}_{y_{0:k}}^{\mathcal{M}, \mathcal{U}_{0:k-1},0,0}) \le M^*,
\end{equation*}
 Dividing by $k+1$, the result follows.
\end{proof}
\section{Proof of Theorem 24} \label{app:flow}
\begin{proof}
When under a replay attack, we have  \cite{Chabukswar2014}
\begin{align}
z_k = z_{k-N} &- \mathcal{P}^{-\frac{1}{2}}C\mathcal{A}^k(\hat{x}_{0|-1} - \hat{x}_{-N|-N-1})  \\ &- \mathcal{P}^{-\frac{1}{2}}C \sum_{j = 0}^{k-1} \mathcal{A}^{k-1-j} B \left(\Delta u_j - \Delta u_{j-N} \right), \nonumber
\end{align}
where $N$ is some unknown, but large delay between the replayed sequence and the true sequence. Thus, under attack $z_k \sim \mathcal{N}(\mu_k,\Sigma_k+I)$ with 
\begin{equation*}
\mu_k  =  \mathcal{P}^{-\frac{1}{2}}C\mathcal{A}^k\hat{x}_{0|-1}  +  \mathcal{P}^{-\frac{1}{2}}C \sum_{j = 0}^{k-1} \mathcal{A}^{k-1-j} B \Delta u_j,
\end{equation*}
\begin{equation*}
\Sigma_k  =   \mathcal{P}^{-\frac{1}{2}}C [\mathcal{A}^k W \mathcal{A}^{k~T} + \sum_{j=0}^{k-1} \mathcal{A}^jB\mathcal{Q}B^T\mathcal{A}^{j~T} ] C^T  \mathcal{P}^{-\frac{1}{2}}.
\end{equation*}
Thus, the KL divergence between $z_k$ under attack and under normal operation is given by
\begin{equation}
D_{KL}(\mathbb{D}_{z_{k}}^{\mathcal{M}, \mathcal{U}_{0:k-1}, \mathcal{U}_{0:k-1}^a, \mathcal{D}_{0:k}^a} ||\mathbb{D}_{z_{k}}^{\mathcal{M}, \mathcal{U}_{0:k-1},0,0}) = \frac{c_k^1 + c_k^2 + c_k^3}{2} \label{eq:KLreplay}
\end{equation}
where 
\begin{equation*}
c_k^1 =  \mu_k^T\mu_k,~~c_k^2 =-\log\det \left( I + \Sigma_k \right ), ~~ c_k^3 = \mbox{tr}(\Sigma_k).
\end{equation*}
From \cite{Mo2014CSM}, it is known that 
\begin{equation}
c_k^2 + c_k^3 \ge 0. \label{eq:ck23}
\end{equation}
Furthermore, by the law of large numbers, we know 
\begin{equation}
 \lim_{T \rightarrow \infty} \frac{1}{T+1} \sum_{k=0}^{T}c_k^1 \overset{a.s.}{\rightarrow}  \mbox{tr} \left(\mathcal{P}^{-1}C\Sigma C^T\right). \label{eq:c1k}
\end{equation}
Using \eqref{eq:KLreplay}, \eqref{eq:ck23} and \eqref{eq:c1k} 
\begin{equation}
 \underset{T \rightarrow \infty}{\lim}  \sum_{k = 0}^T \frac{D_{KL}(\mathbb{D}_{z_{k}}^{\mathcal{M}, \mathcal{U}_{0:k-1}, \mathcal{U}_{0:k-1}^a, \mathcal{D}_{0:k}^a} ||\mathbb{D}_{z_{k}}^{\mathcal{M}, \mathcal{U}_{0:k-1},0,0})}{T+1} \ge \epsilon.
\end{equation}
By Theorem \ref{sumofKLS}, the result immediately follows.
\end{proof}

\bibliographystyle{IEEEtran}
\bibliography{CDC16_infoflow.bib}

\begin{thebibliography}{10}
\providecommand{\url}[1]{#1}
\csname url@samestyle\endcsname
\providecommand{\newblock}{\relax}
\providecommand{\bibinfo}[2]{#2}
\providecommand{\BIBentrySTDinterwordspacing}{\spaceskip=0pt\relax}
\providecommand{\BIBentryALTinterwordstretchfactor}{4}
\providecommand{\BIBentryALTinterwordspacing}{\spaceskip=\fontdimen2\font plus
\BIBentryALTinterwordstretchfactor\fontdimen3\font minus
  \fontdimen4\font\relax}
\providecommand{\BIBforeignlanguage}[2]{{%
\expandafter\ifx\csname l@#1\endcsname\relax
\typeout{** WARNING: IEEEtran.bst: No hyphenation pattern has been}%
\typeout{** loaded for the language `#1'. Using the pattern for}%
\typeout{** the default language instead.}%
\else
\language=\csname l@#1\endcsname
\fi
#2}}
\providecommand{\BIBdecl}{\relax}
\BIBdecl

\bibitem{challengessecurity}
A.~A. C\'{a}rdenas, S.~Amin, and S.~Sastry, ``Research challenges for the
  security of control systems,'' in \emph{HOTSEC'08: Proceedings of the 3rd
  conference on Hot topics in security}.\hskip 1em plus 0.5em minus 0.4em\relax
  Berkeley, CA, USA: USENIX Association, 2008, pp. 1--6.

\bibitem{cardenas2009rethinking}
T.~Cardenas, A. A.and~Roosta and S.~Sastry, ``Rethinking security properties,
  threat models, and the design space in sensor networks: A case study in scada
  systems,'' \emph{Ad Hoc Networks}, vol.~7, no.~8, pp. 1434--1447, 2009.

\bibitem{Langner2013}
\BIBentryALTinterwordspacing
R.~Langner, ``To kill a centrifuge: A technical analysis of what stuxnet's
  creators tried to achieve,'' Langner Communications, Tech. Rep., November
  2013. [Online]. Available:
  \url{www.langner.com/en/wp-content/uploads/2013/11/To-kill-a-centrifuge.pdf}
\BIBentrySTDinterwordspacing

\bibitem{Slay2008}
J.~Slay and M.~Miller, ``Lessons learned from the maroochy water breach,'' in
  \emph{Critical Infrastructure Protection}.\hskip 1em plus 0.5em minus
  0.4em\relax Springer US, 2008, pp. 73--82.

\bibitem{jones_report}
H.~L. Jones, ``Failure detection in linear systems,'' Ph.D. dissertation,
  M.I.T., Cambridge, Massachusetts, 1973.

\bibitem{failuredetection}
A.~S. Willsky, ``A survey of design methods for failure detection in dynamic
  systems,'' \emph{Automatica}, vol.~12, pp. 601--611, Nov 1976.

\bibitem{Mo_robust_detection}
Y.~Mo, J.~Hespanha, and B.~Sinopoli, ``Robust detection in the presence of
  integrity attacks,'' in \emph{American Control Conference (ACC), 2012}, June
  2012, pp. 3541--3546.

\bibitem{PasqualettiAttack}
F.~Pasqualetti, F.~Dorfler, and F.~Bullo, ``Attack detection and identification
  in cyber-physical systems,'' \emph{Automatic Control, IEEE Transactions on},
  vol.~58, no.~11, pp. 2715--2729, Nov 2013.

\bibitem{sundaram_wirelesscontrol}
S.~Sundaram, M.~Pajic, C.~Hadjicostis, R.~Mangharam, and G.~J. Pappas, ``The
  wireless control network: monitoring for malicious behavior,'' in \emph{IEEE
  Conference on Decision and Contro}, Atlanta, GA, Dec 2010.

\bibitem{DenningD77}
\BIBentryALTinterwordspacing
D.~E. Denning and P.~J. Denning, ``Certification of programs for secure
  information flow,'' \emph{Commun. {ACM}}, vol.~20, no.~7, pp. 504--513, 1977.
  [Online]. Available: \url{http://doi.acm.org/10.1145/359636.359712}
\BIBentrySTDinterwordspacing

\bibitem{Bai2015}
C.-Z. Bai, F.~Pasqualetti, and V.~Gupta, ``Security in stochastic control
  systems: Fundamental limitations and performance bounds,'' in \emph{American
  Control Conference (ACC), 2015}, June 2015.

\bibitem{moscs10security}
Y.~Mo and B.~Sinopoli, ``False data injection attacks in control systems,'' in
  \emph{First Workshop on Secure Control Systems}, Stockholm, Sweden, April
  2010.

\bibitem{mo2012integrity}
------, ``Integrity attacks on cyber-physical systems,'' in \emph{Proceedings
  of the 1st international conference on High Confidence Networked
  Systems}.\hskip 1em plus 0.5em minus 0.4em\relax ACM, 2012, pp. 47--54.

\bibitem{Tex2012}
A.~Teixeira, I.~Shames, H.~Sandberg, and K.~Johansson, ``Revealing stealthy
  attacks in control systems,'' in \emph{Communication, Control, and Computing
  (Allerton), 2012 50th Annual Allerton Conference on}, Oct 2012, pp.
  1806--1813.

\bibitem{miao_coding}
F.~Miao, Q.~Zhu, M.~Pajic, and G.~Pappas, ``Coding sensor outputs for injection
  attacks detection,'' in \emph{Decision and Control (CDC), 2014 IEEE 53rd
  Annual Conference on}, Dec 2014, pp. 5776--5781.

\bibitem{WeerakkodyCDC2015}
S.~Weerakkody and S.~B., ``Detecting integrity attacks on control systems using
  a moving target approach,'' in \emph{Submitted to Decision and Control (CDC),
  2015 IEEE 54th Annual Conference on}, Dec 2015.

\bibitem{Mo2009R}
Y.~Mo and B.~Sinopoli, ``Secure control against replay attacks,'' in
  \emph{Communication, Control, and Computing, 2009. Allerton 2009. 47th Annual
  Allerton Conference on}, Sept 2009, pp. 911--918.

\bibitem{Chabukswar2014}
Y.~Mo, R.~Chabukswar, and B.~Sinopoli, ``Detecting integrity attacks on scada
  systems,'' \emph{Control Systems Technology, IEEE Transactions on}, vol.~22,
  no.~4, pp. 1396--1407, July 2014.

\bibitem{Weerakkody2014}
S.~Weerakkody, Y.~Mo, and B.~Sinopoli, ``Detecting integrity attacks on control
  systems using robust physical watermarking,'' in \emph{Decision and Control
  (CDC), 2014 IEEE 53rd Annual Conference on}, Dec 2014, pp. 3757--3764.

\bibitem{Mo2014CSM}
Y.~Mo, S.~Weerakkody, and B.~Sinopoli, ``Physical authentication of control
  systems: Designing watermarked control inputs to detect counterfeit sensor
  outputs,'' \emph{Control Systems, IEEE}, vol.~35, no.~1, pp. 93--109, Feb
  2015.

\bibitem{miao2013}
F.~Miao, M.~Pajic, and G.~Pappas, ``Stochastic game approach for replay attack
  detection,'' in \emph{Decision and Control (CDC), 2013 IEEE 52nd Annual
  Conference on}, Dec 2013, pp. 1854--1859.

\bibitem{Cover:2006:EIT:1146355}
T.~M. Cover and J.~A. Thomas, \emph{Elements of Information Theory (Wiley
  Series in Telecommunications and Signal Processing)}.\hskip 1em plus 0.5em
  minus 0.4em\relax Wiley-Interscience, 2006.

\bibitem{kullback1968information}
S.~Kullback, \emph{Information theory and statistics}.\hskip 1em plus 0.5em
  minus 0.4em\relax Courier Corporation, 1968.

\bibitem{GoguenM82}
J.~A. Goguen and J.~Meseguer, ``Security policies and security models,'' in
  \emph{{IEEE} Symposium on Security and Privacy}, 1982, pp. 11--20.

\bibitem{VolpanoS99}
D.~M. Volpano and G.~Smith, ``Probabilistic noninterference in a concurrent
  language,'' \emph{Journal of Computer Security}, vol.~7, no.~1, 1999.

\bibitem{Smith09}
G.~Smith, ``On the foundations of quantitative information flow,'' in
  \emph{Foundations of Software Science and Computational Structures, 12th
  International Conference, {FOSSACS} 2009, Held as Part of the Joint European
  Conferences on Theory and Practice of Software, {ETAPS} 2009, York, UK, March
  22-29, 2009. Proceedings}, 2009, pp. 288--302.

\bibitem{Kalman1960}
R.~E. Kalman, ``A new approach to linear filtering and prediction problems,''
  \emph{Journal of Fluids Engineering}, vol.~82, no.~1, pp. 35--45, 1960.

\bibitem{mehra1971innovations}
R.~K. Mehra and J.~Peschon, ``An innovations approach to fault detection and
  diagnosis in dynamic systems,'' \emph{Automatica}, vol.~7, no.~5, pp.
  637--640, 1971.

\bibitem{Willsky1974}
A.~S. Willsky, ``On the invertibility of linear systems,'' \emph{IEEE
  Transactions on Automatic Control}, vol.~19, no.~3, pp. 272--274, 1974.

\end{thebibliography}
\end{document}